%% file: main.tex
\documentclass{vldb}
\usepackage{times,amsmath,epsfig}

\usepackage{sty/term}       
\usepackage{sty/local}      

\usepackage{enumitem}

\newdef{definition}{Definition}
\newdef{example}{Example}
\newdef{lemma}{Lemma}

\begin{document}

\title{Errata Note: Discovering Order Dependencies through Order Compatibility}

\input{autV}





\maketitle
\begin{abstract}
A number of extensions
to the classical notion of functional dependencies have been proposed
to express and enforce application semantics.
One of these extensions is that of order dependencies (ODs),
which express rules involving order.
The article entitled 
``Discovering Order Dependencies through Order Compatibility''
by Consonni et al.,
published in the EDBT conference proceedings in March 2019,
investigates the OD discovery problem.
They claim to prove that their OD discovery algorithm,
OCDDISCOVER,
is \emph{complete},
as well as being significantly more efficient in practice
than the state-of-the-art.
They further claim that the implementation
of the existing FASTOD algorithm
(ours)---%
we shared our code base with the authors---%
which they benchmark against is flawed,
as OCDDISCOVER and FASTOD report different sets of ODs
over the same data sets.

In this rebuttal,
we show that their claim of completeness is, in fact, \emph{not} true.
Built upon their incorrect claim,
OCDDISCOVER's pruning rules are overly aggressive,
and prune parts of the search space that contain legitimate ODs.
This is the reason their approach appears to be ``faster'' in practice. Finally, we show that Consonni et al. misinterpret our set-based canonical form for ODs, leading to an incorrect claim that our FASTOD implementation has an error. 


\end{abstract}


\input{introduction}
\input{found}
\input{comp}
\input{exp}

\input{concl}
%






\bibliographystyle{abbrv}
\bibliography{main}

\end{document}

%% file: autV.tex
\numberofauthors{3} 
\author{
\alignauthor
Parke Godfrey\\
       \affaddr{York University}\\
       \affaddr{Ontario, Canada}\\
       \email{godfrey@yorku.ca}
\alignauthor
Lukasz Golab\\
       \affaddr{University of Waterloo}\\
       \affaddr{Ontario, Canada}\\
       \email{lgolab@uwaterloo.ca}
\alignauthor
Mehdi Kargar\\
       \affaddr{Ryerson University}\\
       \affaddr{Ontario, Canada}\\
       \email{kargar@ryerson.ca}
\and  
\alignauthor
Divesh Srivastava\\
       \affaddr{AT\&T Labs-Research}\\
       \affaddr{Ontario, Canada}\\
       \email{divesh@research.att.com}
\alignauthor
Jaroslaw Szlichta\\
       \affaddr{Ontario Tech University}\\
       \affaddr{Ontario, Canada}\\
       \email{jarek@uoit.ca}
}

\makeatletter
\def\@copyrightspace{\relax}
\makeatother

%% file: introduction.tex
\section{Introduction}\label{section:intro}

Integrity constraints specify the intended semantics of dataset attribu\-tes.
They are commonly used in a number of application areas,
such as schema design,
data integration,
data cleaning, and
query optimization~\cite{CFG07}.
Past work focused primarily on \emph{functional dependencies} (FDs).
In recent years,
several extensions to the notion of an FD have been studied,
including that of \emph{order dependencies} (ODs)
\cite{%
    CSM19,
    LN16,
    SGG+17,
    SGG+18,
    Szl2012,
    SGG+2013
}.
FDs cannot capture relationships among attributes
with naturally \emph{ordered} domains,
such as over timestamps, numbers, and strings,
which are common in business data \cite{Szli2012}.
For example,
consider Table~\ref{table:taxes},
which shows employee tax records in which $\AB{tax}$ is calculated
as a $\AB{percentage}$ of $\AB{salary}$.
Both $\AB{tax}$ and $\AB{percentage}$ increase with $\AB{salary}$.

Order dependencies naturally express such semantics.
For a second example from Table~\ref{table:taxes},
the OD $\AB{salary}$ \emph{orders} $\AB{group}$, $\AB{subgroup}$ holds.
When the table is sorted by $\AB{salary}$,
it is also then sorted by $\AB{group}$
(with ties broken by $\AB{subgroup}$).
However,
$\AB{salary}$ \emph{orders} $\AB{subgroup}$, $\AB{group}$ does \emph{not} hold. This illustrates that the \emph{order} of attributes matters.

The theory of order dependency subsumes that of functional dependency.
Any FD can be \emph{mapped} to an equivalent OD
by prefixing the left-hand-side attributes
onto the right-hand side
\cite{Szl2012,SGG+2013}.
For example,
if $\AB{salary}$ \emph{functionally determines} $\AB{tax}$,
then $\AB{salary}$ \emph{orders} $\AB{salary}, \AB{tax}$.

The purpose of this article is to
refute the following claims
in Consonni et al.~
\cite{CSM19}.

\begin{enumerate}[nolistsep,leftmargin=*]
\item \label{CLAIM minimality}
    The authors present a definition of \emph{minimality}
    for order compatibility dependencies (OCDs).
    An OCD is a more specific form of order dependency
    in which two lists of attributes order each other,
    when taken together
    \cite{Szl2012}. 
    Consonni et al.~\cite{CSM19} claim
    that their definition of minimality is \emph{complete};
    that is,
    from it,
    one can recover all valid OCDs
    that hold over a given table.
    
\item \label{CLAIM completeness}
    Given their definition of minimal OCDs,
    Consonni et al.~\cite{CSM19} propose an algorithm
    to \emph{discover} ODs via OCDs,
    which has factorial complexity in the number of attributes.
    They claim to prove that their algorithm produces
    a canonically \emph{complete} set of ODs.
    (That is,
    a \emph{minimal} set of ODs
    with respect to their definition,
    from which all the ODs which hold over the data
    could purportedly be inferred.)
    
\item \label{CLAIM efficiency}
    The authors claim that their experimental evaluation illustrates
    an \emph{implementation error}
    in our implementation of our OD discovery algorithm (FASTOD)
    \cite{SGG+17,SGG+18},
    which leads to ours discovering many additional---%
    and, purportedly, incorrect---%
    depende\-ncies.
    In spite of this claim of an ``implementation error''
    in the FASTOD implementation that we provided them,
    they support via benchmark experiments
    that their algorithm, OCDDISCOVER, outperforms
    our algorithm, FASTOD.
\end{enumerate}

\noindent 
We show that each of these three claims is incorrect,
in turn.

\begin{enumerate}[nolistsep,leftmargin=*]
\item \label{REFUTE minimality}
    The definition of minimality
    in Consonni et al.~\cite{CSM19}---%
    insofar as its intended purpose is a \emph{canonical} form---%
    is incorrect.
    Their ``canonical'' form does not allow
    for the inference of \emph{all} OCDs.
    It misses an important subclass of OCDs (and, respectively, ODs),
    any dependency which has a common prefix on the \emph{left} and \emph{right}
    (that is, repeated attributes at the beginning of the dependency).

\item \label{REFUTE completeness}
    The claim of completeness of the OD discovery algorithm
    in Consonni et al.~
    \cite{CSM19}
    is \emph{incorrect},
    as it relies upon their incorrect notion of ``minimal'' OCDs.
    Their conjecture that their algorithm is complete is incorrect;
    it is incomplete.

\item \label{REFUTE efficiency}
    Consonni et al.~\cite{CSM19} misinterpret our set-based canonical form
    for ODs
    \cite{SGG+17,SGG+18}
    (which is equivalent to the list-based canonical form for ODs).
    This leads the authors to confuse set-based OCDs with ODs.
    Their claim that our implementation has an error arises from this,
    and their belief that their approach is complete.
    Consonni et al.~\cite{CSM19} conclude that their algorithm is faster
    in practice,
    despite being significantly worse in asymptotic complexity.
    This arises in their benchmark experiments,
    however,
    due to the fact that their algorithm is incomplete,
    
\end{enumerate}

\begin{table}[t]
\begin{small}
\caption{Table with employee information.}\label{table:taxes}
    \begin{tabular}{|p{0.35cm}|p{0.35cm}|p{0.35cm}|p{0.65cm}|p{0.35cm}|p{0.5cm}|p{0.55cm}|p{0.58cm}|p{0.4cm}|p{0.55cm}|}
        \hline
        \A{\#} & \A{ID} & \A{yr} & \A{posit} & \A{bin} & \A{sal} & \A{perc} & \A{tax} & \A{grp} & \A{subg} \Tstrut\\
        \hline \hline
        $\tup{t1}$ & 10 &  19 &  secr & 1 & 5K &  20\% &  1K & A & III \Tstrut\\
        \hline
        $\tup{t2}$ & 11 &  19 &  mngr &  2 & 8K &  25\% &  2K & C & II \Tstrut\\
        \hline
        $\tup{t3}$ & 12 &  19 &  direct &  3 & 10K &  30\% &  3K & D & I \Tstrut\\
        \hline
        \hline
        $\tup{t4}$ & 10 &  18 &  secr &  1 & 4.5K &  20\% &  0.9K & A & III \Tstrut\\
        \hline
        $\tup{t5}$ & 11 &  18 &  mngr &  2 & 6K &  25\% &  1.5K & C & I \Tstrut\\
        \hline
        $\tup{t6}$ & 12 &  18 &  direct &  3 & 8K &  25\% &  2K & C & II \Tstrut\\
        \hline
    \end{tabular}
\end{small}
 \centering
\end{table}

In Section~\ref{sec:foundations}
we provide basic definitions and canonical forms for ODs.
In Section~\ref{section:completness},
we analyze the completeness of OD discovery.
In Section~\ref{section:evaluation},
we discuss the experimental evaluation conducted
by Consonni et al.~\cite{CSM19}.
We conclude in Section~\ref{section:conc}.

%% file: found.tex
\section{Foundations}\label{sec:foundations}
\subsection{Background}

We use the notational conventions in Table~\ref{table:Notation}.
Next, we provide a summary of the relevant definitions. The operator `\orelOP[\lst{X}]' defines a {\em weak total order} over any set of tuples, where $\lst{X}$ denotes a list of attributes. Unless otherwise specified, numbers are ordered numerically, strings are ordered lexicographically and dates are ordered chronologically.

\begin{definition}~\cite{SGG+17,SGG+18}
\label{definition:operatorle}
Let \lst{X}\ be a list of attributes. For two tuples \tup{t}\ and \tup{s}, $\set{X} \in \R{R}$,
\(\tup{t}\orel[\lst{X}]\tup{s}\) if%
\footnote{
    By some conventions,
    ``\emph{iff}''---%
    ``if and only if---%
    would be used here.
    The intent, in any case,
    is that this defines completely the notion.
}
\begin{itemize}[nolistsep,leftmargin=*]
\item[--]
    \(\lst{X} = \Lst{}\);
    {\em or}
\item[--]
    \(\lst{X} = \Lst{\A{A}}[\lst{T}]\)
    and \(\tup{t}[\A{A}] < \tup{s}[\A{A}]\); or
\item[--]
    \(\lst{X} = \Lst{\A{A}}[\lst{T}]\),
    \(\tup{t}[\A{A}] = \tup{s}[\A{A}]\),
    and
    \(\tup{t}\orel[\lst{T}]\tup{s}\).
\end{itemize}
Let \(\tup{t}\orelStrict[\lst{X}]\tup{s}\) if
    \(\tup{t}\orel[\lst{X}]\tup{s}\) but
    \(\tup{s}\notOrel[\lst{X}]\tup{t}\).
\end{definition}

Next, we define order dependencies.

\begin{definition}~\cite{CSM19,LN16,SGG+17,SGG+18,Szl2012,SGG+2013}
\label{definition:OrderDependency}
Let \lst{X}\ and \lst{Y}\ be lists of attributes over a relation schema $\R{R}$. 
%
Table \T{r}\ over $\R{R}$ {\em satisfies} an OD \orders{\lst{X}}{\lst{Y}}\ (\(\T{r} \models \orders{\lst{X}}{\lst{Y}}\)), read as \lst{X}\ {\em orders} \lst{Y},
if for all \(\tup{t}, \tup{s} \in \T{r}\), \(\tup{t}\orel[\lst{X}]\tup{s}\) implies \(\tup{t}\orel[\lst{Y}]\tup{s}\).
\orders{\lst{X}}{\lst{Y}}\ is said to {\em hold} for \R{R}\ (\(\R{R} \models \orders{\lst{X}}{\lst{Y}}\)) if, for each admissible table instance \T{r}\ of \R{R}, table \T{r}\ satisfies \orders{\lst{X}}{\lst{Y}}.
\orders{\lst{X}}{\lst{Y}}\ is {\em trivial} if, for all \T{r}, \(\T{r} \models \orders{\lst{X}}{\lst{Y}}\). \orderEquiv{\lst{X}}{\lst{Y}}, read as \lst{X}\ and \lst{Y}\ are {\em order equivalent}, if $\orders{\lst{X}}{\lst{Y}}$ and $\orders{\lst{Y}}{\lst{X}}$.
\end{definition}

The OD $\orders{\lst{X}}{\lst{Y}}$ means that $\lst{Y}$ values are monotonically non-decreasing wrt $\lst{X}$ values. Thus, if a list of tuples is ordered by $\lst{X}$, then it is also ordered by $\lst{Y}$, but not necessarily vice versa.

\begin{example}\label{example:taxes}
Consider Table~\ref{table:taxes} in which tax is calculated as a percentage of salary, and tax groups and subgroups are based on salary.  Tax, percentage and group are not decreasing with salary.  Furthermore, within the same group, subgroup are not decreasing with salary.  Finally, within the same year, bin increases with salary.  Thus, the following order dependencies hold in that table:
    $\orders{ \Lst{ \AB{salary} } }{ \Lst{ \AB{tax} } }$,
     $\orders{ \Lst{ \AB{salary} } }{ \Lst{ \AB{percentage} } }$,
    $\orders{ \Lst{ \AB{salary} } },
        { \Lst{ \AB{group}, \AB{subgroup} } }$ and 
    $\orders{ \Lst{ \AB{year, salary} } }{ \Lst{ \AB{year}, \AB{bin} } }$.
\end{example}

\begin{definition}~\cite{Szl2012,SGG+2013}
    \label{definition:orderCompatible}
Two order specifications \lst{X}\ and \lst{Y}\ are {\em order
compatible}, denoted as \(\lst{X} \sim \lst{Y}\), if
\orderEquiv{\lst{XY}}{\lst{YX}}. ODs in the form of $\simular{\lst{X}}{\lst{Y}}$ are called order compatible dependencies (OCDs)
\end{definition}

The empty list of attributes (i.e., $\emptyLst{}$) is order compatible with {\em any} list of attributes. There is a strong relationship between ODs and FDs. Any OD implies an FD, modulo lists and sets, however, not vice versa.

\begin{table}[t]
\caption{Notational conventions.}
\label{table:Notation}
\centering
\begin{normalsize}
\begin{itemize}
    \item \textbf{Relations.}
            \R{R}\ denotes a {\em relation schema} and
            \T{r}\ denotes a specific {\em table} instance.
            Letters from the beginning of the alphabet, \A{A}, \A{B}\ and \A{C},
            denote single \emph{attributes}.
            Additionally, \tup{t}\ and \tup{s}\
            denote \emph{tuples}, and
            \proj{t}{\A{A}}\ denotes
            the value of an attribute \A{A}\ in a tuple \tup{t}.
    \item \textbf{Sets.}
            Letters from the end of the alphabet, \set{X}, $\set{Y}$ and \set{Z}, denote \emph{sets} of attributes.
            Also, \proj{t}{\set{X}}\ denotes
            the \emph{projection} of a tuple \emph{t} on $\set{X}$.
            $\set{X}\set{Y}$ is shorthand for $\set{X} \cup \set{Y}$. The empty set is denoted as $\emptySet{}$.

    \item \textbf{Lists.}
            \lst{X}, \lst{Y}\ and \lst{Z}\ denote \emph{lists}. The empty list is represented as $\emptyLst{}$.
            List \Lst{\A{A}, \A{B}, \A{C}}\ denotes an explicit list.
            \Lst{\A{A}}[\lst{T}]\
            denotes a list with the \emph{head} \A{A}\
            and the \emph{tail} \lst{T}.
            \lst{XY}\ is shorthand for \lst{X}\ concatenate \lst{Y}.
            Set \set{X}\ denotes the {\em set} of elements
            in {\em list} \lst{X}. $\lst{X}^{p}$ denotes any arbitrary permutation of list $\lst{X}$ or set $\set{X}$. Given a set of attributes $\set{X}$, for brevity, we state $\forall i$, $\A{X}_{i}$ to indicate indices $[1, ..., i]$ that have valid ranges ($i \leq |\set{X}|$).
\end{itemize}
\end{normalsize}
\vskip -0.2cm
\line(1,0){240}
\end{table}

\begin{lemma}~\cite{Szl2012,SGG+2013}
\emph{
\label{theorem:relationship}
\emph{
If $\R{R}$ $\models$ \orders{\lst{X}}{\lst{Y}}\ (OD), then $\R{R}$ $\models$ $\set{X} \rightarrow \set{Y}$ (FD).
}
}
\end{lemma}

Also, there is a correspondence between FDs and ODs.
%
\begin{theorem}\emph{~\cite{Szl2012,SGG+2013}} \label{theorem:correspondence}
\emph{
$\R{R}$ $\models$ $\set{X} \rightarrow \set{Y}$ \emph{iff} \orders{\lst{X}}{\lst{XY}}, for any list \lst{X}\ over the attributes of \set{X}\ and any list \lst{Y}\ over the attributes of \set{Y}.
}
\end{theorem}

ODs can be violated in two ways. 

\begin{theorem}\emph{~\cite{Szl2012,SGG+2013}}
\label{theorem:orderdependency}
\emph{
$\R{R}$ $\models$ \orders{\lst{X}}{\lst{Y}}\ (OD) \emph{iff}
$\R{R}$ $\models$ $\orders{\lst{X}}{\lst{XY}}$ (FD) and $\simular{\lst{X}}{\lst{Y}}$ (OCD).
}
\end{theorem}

We are now ready to explain the two sources of OD violations: \emph{splits} and \emph{swaps}~\cite{Szl2012,SGG+2013}.  An OD $\orders{\lst{X}}{\lst{Y}}$ can be violated in two ways, as per Theorem~\ref{theorem:orderdependency}. 
\begin{definition}~\cite{Szl2012,SGG+2013}
A \emph{split} wrt an OD $\orders{\lst{X}}{\lst{XY}}$ (FD) is a pair of tuples \emph{s} and \emph{t} such that $\emph{s}_{\set{X}}$ $=$ $\emph{t}_{\set{X}}$ but
$\emph{s}_{\set{Y}}$ $\not =$ $\emph{t}_{\set{Y}}$.\label{definition:split}
\end{definition}

\begin{definition}~\cite{Szl2012,SGG+2013}
A \emph{swap} wrt $\simular{\lst{X}}{\lst{Y}}$ (OCD) is a pair of tuples \emph{s} and \emph{t} such that $\emph{s} \prec_{\lst{X}} \emph{t}, \mbox{but}\ \emph{t} \prec_{\lst{Y}} \emph{s}$.
\label{definition:swap}
\end{definition}

\begin{example}\label{example:swap} In Table~\ref{table:taxes}, there are three splits with respect to the OD $\orders{ \Lst{ \AB{position} } }{ \Lst{ \AB{position}, \AB{salary} } }$ because $\AB{position}$ does not functionally determine $\AB{salary}$.  The violating tuple pairs are t1 and t4, t2 and t5, and t3 and t6.
There is a swap wrt $\simular{ \Lst{ \AB{salary} } }{ \Lst{ \AB{subgroup} } }$, e.g., over pair of tuples $\tup{t1}$ and $\tup{t2}$.
\end{example}

\subsection{Canonical Forms}

Consonni et al.~\cite{CSM19} use a native list-based canonical form, which is based on decomposing an OD into a FD and an OCD~\cite{Szl2012,SGG+2013}. Recall that based on Theorem~\ref{theorem:orderdependency} ``OD = FD + OCD", as $\orders{\lst{X}}{\lst{Y}}$ \emph{iff} $\orders{\lst{X}}{\lst{XY}}$ (FD) and $\simular{\lst{X}}{\lst{Y}}$ (OCD). The authors exploit this relationship to guide their discovery algorithm through order compatibility. Since they use a list-based representation for ODs, this leads to factorial complexity of OD discovery in the number of attributes. 

Expressing ODs in a natural way relies on lists of attributes, as in the SQL order-by statement. One might well wonder whether lists are inherently necessary.  We provide a polynomial \emph{mapping} of list-based ODs into \emph{equivalent} set-based canonical ODs~\cite{SGG+17,SGG+18}. The mapping allows us to develop an OD discovery algorithm that traverses a much smaller set-containment lattice (to identify candidates for ODs) rather than the list-containment lattice used in Consonni et al.~\cite{CSM19}.

Two tuples, $\tup{t}$ and $\tup{s}$, are \emph{equivalent} over a set of attributes $\set{X}$ if $\proj{t}{\set{X}}$ = $\proj{s}{\set{X}}$. An attribute set $\set{X}$ partitions tuples into \emph{equivalence classes}~\cite{HKP98}. We denote the \emph{equivalence class} of a tuple $\tup{t} \in \T{r}$ over a set $\set{X}$ as 
$\set{E}(\tup{t}_{\set{X}})$, i.e., $\set{E}(\tup{t}_{\set{X}})$ = $\{ \tup{s} \in \T{r}$ $|$ $\proj{s}{\set{X}}$ = $\proj{t}{\set{X}} \}$.
A \emph{partition} of $\T{r}$ over $\set{X}$ is the set of equivalence classes, $\Pi_{\set{X}}$ = $\{ \set{E}(\tup{t}_{\set{X}})$ $|$ $\tup{t} \in \T{r} \}$. For instance, in Table~\ref{table:taxes}, $\set{E}(\tup{t1}_{\brac{\A{year}}})$ = $\set{E}(\tup{t2}_{\brac{\A{year}}})$ = $\set{E}(\tup{t3}_{\brac{\A{year}}})$ = $\{ \tup{t1},  \tup{t2}, \tup{t3}\}$ and $\Pi_{\A{year}}$ = $\{ \{ \tup{t1},  \tup{t2}, \tup{t3}\}, \{ \tup{t4},  \tup{t5}, \tup{t6}\} \}$.

We now present a set-based \emph{canonical form} for ODs.

\begin{definition}~\cite{SGG+17,SGG+18}
An attribute $\A{A}$ is a \emph{constant} within each equivalence class over $\set{X}$, denoted as $\set{X}$: $\ordersCst{\emptyLst{}}{\A{A}}$, if $\orders{\lst{X}^{p}}{\lst{X}^{p} \A{A} }$.  Furthermore, two attributes,  $\A{A}$ and $\A{B}$, are order-compatible 
within each equivalence class wrt $\set{X}$, denoted as $\set{X}$: $\simular{\A{A}}{\A{B}}$, if $\simular{\lst{X}^{p} \A{A}}{\lst{X}^{p} \A{B} }$.  ODs of the form of $\set{X}$:$\ordersCst{\emptyLst{}}{\A{A}}$ and $\set{X}$: $\simular{ \A{A} }{ \A{B} }$ are called \emph{(set-based) canonical} ODs, and the set $\set{X}$ is called a \emph{context}.
\label{def:canonicalODs}
\end{definition}

\begin{example}\label{example:ODs}
In Table~\ref{table:taxes}, an attribute $\AB{bin}$ is a constant in the context of $\AB{position}$ $($$\AB{posit}$$)$, written as
    $\brac{\AB{position}}$$:$$\ordersCst{\emptyLst{}}{\AB{bin}}$.
This is because
    $\set{E}(\tup{t1}_{\brac{\AB{position}}})$ $\models$$\ordersCst{\emptyLst{}}{\AB{bin}}$,
    $\set{E}(\tup{t2}_{\brac{\AB{position}}})$ $\models$$\ordersCst{\emptyLst{}}{\AB{bin}}$ and
    $\set{E}(\tup{t3}_{\brac{\AB{position}}})$ $\models$$\ordersCst{\emptyLst{}}{\AB{bin}}$.
Also, there is no swap between $\AB{bin}$ and $\AB{salary}$ in the context of $\AB{year}$, i.e.,
    $\brac{\AB{year}}$$:$ $\simular{\AB{bin}}{\AB{salary}}$.
This is because
    $\set{E}(\tup{t1}_{\brac{\AB{year}}})$ $\models$ $\simular{\AB{bin}}{\AB{salary}}$ and
    $\set{E}(\tup{t4}_{\brac{\AB{year}}})$ $\models$ $\simular{\AB{bin}}{\AB{salary}}$.
\end{example}

Based on Theorem~\ref{thm:canSplitOD} and Theorem~\ref{thm:canSwapOD},  list-based ODs in the form of FDs and OCDs, respectively, can be mapped into equivalent set-based ODs. 

\begin{theorem}\emph{~\cite{SGG+17,SGG+18}}
\emph{
$\R{R}$ $\models$ $\orders{\lst{X}}{\lst{X}\lst{Y}}$ \emph{iff} $\forall \A{A} \in \lst{Y}$, $\R{R}$ $\models$ $\set{X}$$:$$\ordersCst{\emptyLst{}}{\A{A}}$.
\label{thm:canSplitOD}
}
\end{theorem}

\begin{theorem}\emph{~\cite{SGG+17,SGG+18}}
\emph{
$\R{R}$ $\models$ $\simular{\lst{X}}{\lst{Y}}$ \mbox{\emph{iff}} $\forall i,j$, $\R{R}$ $\models$ $\{ \A{X}_{1}, .., \A{X}_{i-1},$ $\A{Y}_{1}, ..,  \A{Y}_{j-1} \}$$:$ $\simular{\A{X}_{i}}{\A{Y}_{j}}$.
\label{thm:canSwapOD}
}
\end{theorem}

A list-based OD can be mapped into an equivalent set of set-based ODs via a polynomial mapping.

\begin{theorem}\emph{~\cite{SGG+17,SGG+18}}
\emph{
$\R{R}$ $\models$ $\orders{\lst{X}}{\lst{Y}}$ \emph{iff} $\forall \A{A} \in \lst{Y}$, $\R{R}$ $\models$ $\set{X}$$:$$\ordersCst{\emptyLst{}}{\A{A}}$ and $\forall i,j$, $\R{R}$ $\models$ $\{ \A{X}_{1}$, .., $\A{X}_{i-1}$, $\A{Y}_{1}$, ..,  $\A{Y}_{j-1} \}$$:$ $\simular{\A{X}_{i}}{\A{Y}_{j}}$.
\label{thm:canonicalODs}
}
\end{theorem}


\begin{example}
An OD $\orders{\Lst{\A{A}\A{B}}}{\Lst{\A{C}\A{D}}}$ can be mapped into the following equivalent canonical ODs:
    $\brac{\A{A}, \A{B}}$$:$$\ordersCst{\emptyLst{}}{\A{C}}$,
    $\brac{\A{A}, \A{B}}$$:$$\ordersCst{\emptyLst{}}{\A{D}}$,
    $\{\}$: $\simular{\A{A}}{\A{C}}$,
    $\{ \A{A} \}$: $\simular{\A{B}}{\A{C}}$,
    $\{ \A{C} \}$: $\simular{\A{A}}{\A{D}}$,
    $\{ \A{A}, \A{C} \}$: $\simular{\A{B}}{\A{D}}$.
\end{example}

%% file: comp.tex
\section{Completeness Analysis}\label{section:completness}

While the theoretical search space
for FASTOD~\cite{SGG+17,SGG+18} is $O(2^{|\R{R}|})$,
the search space for OCDDISCOVER~\cite{CSM19} is $O(|\R{R}|!)$, which is much larger as it traverses a lattice of attribute \emph{permutations}
(where $|\R{R}|$ denotes the number of attributes
over a relational schema $\R{R}$).
To mitigate the factorial complexity,
the list-based algorithm
in Consonni et al.~\cite{CSM19} uses
pruning rules.
We show that,
despite the authors' claim that their approach discovers
a canonically complete set of ODs,
their pruning rules lead to \emph{incompleteness}. 

Section 3 in Consonni et al.~\cite{CSM19} addresses
their completeness ``proof'' for their OD discovery algorithm.
The authors introduce a notion of \emph{minimality}
of a set of dependencies which is incorrect.
Herein,
a set of dependencies is called \emph{minimal}---%
as it is in previous work on FDs and ODs
\cite{HKP98,SGG+17,SGG+18}---%
if \emph{all} dependencies that logically hold
over a relation schema $\R{R}$ can be inferred
from this minimal (canonical) set of dependencies.%
\footnote{In some previous work~\cite{Ber1979}, minimal dependencies $\set{M}$ also satisfy an additional condition that  
that no proper subset of $\set{M}$ can be used to infer all dependencies.
}
%
That is, a set of dependencies $\set{M}$ is \emph{minimal}
over a table $\R{r}$,
if 
$\{ \orders{\lst{X}}{\lst{Y}}$ $|$ $\set{M} \models \orders{\lst{X}}{\lst{Y}} \}$ is equivalent to 
$\{ \orders{\lst{X}}{\lst{Y}}$ $|$ $\R{r} \models \orders{\lst{X}}{\lst{Y}} \}$.

Thus,
one should be able to infer
from a minimal set of dependencies
via the inference rules (axioms), $\set{I}$,
\emph{all} the dependencies that are valid
over the given instance of the table.
That is,
$\{ \orders{\lst{X}}{\lst{Y}}$ $|$ $\set{M} \vdash_{I} \orders{\lst{X}}{\lst{Y}} \}$ is equal to 
$\{ \orders{\lst{X}}{\lst{Y}}$ $|$ $\R{r} \models \orders{\lst{X}}{\lst{Y}} \}$.
Consonni et al.~\cite{CSM19} use
the set of \emph{sound} and \emph{complete} OD inference rules,
$\set{I}$,
from \cite{Szli2012,SGG+2013}.

Pruning applied by a dependency discovery algorithm, thus,
must respect minimality.
This allows for the \emph{implicit} discovery
of the full set of valid dependencies,
and thus be deemed \emph{complete}.

In \cite{CSM19},
an attribute list is minimal
if it has no embedded order dependency
(the list of attributes is the shortest possible).

\begin{definition}~\cite{CSM19}
An attribute list $\lst{X}$ is \emph{minimal}
if there is no other list of attributes $\lst{X}'$
such that:
\begin{itemize}[nolistsep,leftmargin=*]
    \item $\lst{X}'$ is smaller than $\lst{X}$, and
    \item $\orderEquiv{\lst{X}}{\lst{X}'}$
\end{itemize}
\end{definition}

\begin{example}
$\Lst{ \A{A}, \A{B},\A{A} }$ is \emph{not} minimal as $\orderEquiv{\Lst{ \A{A}, \A{B},\A{A} }}{\Lst{ \A{A}, \A{B} }}$. 
\end{example}

It follows then
that an OCD is minimal in \cite{CSM19}
\emph{if and only if}
there are no repeated attributes in the OCD.
That is,
there are no repeated attributes
within the \emph{left} or within the \emph{right} list
of the minimal OCD, as each is a minimal attribute list,
\emph{and} there is no repeated attribute
between \emph{left} and \emph{right}.

\begin{definition}\label{minimalOCD}
\cite{CSM19}
An OCD $\simular{\lst{X}}{\lst{Y}}$ is \emph{minimal}
if
\begin{itemize}[nolistsep,leftmargin=*]
\item $\lst{X}$ and $\lst{Y}$ are minimal attribute lists
    \emph{and}
\item $\set{X}$ $\cap$ $\set{Y}$ $=$ $\emptyset$.
\end{itemize}
\end{definition}

Definition \ref{minimalOCD} \cite{CSM19}
of \emph{minimality}
with no permitted repeated attributes
is at the heart of their incompleteness problem,
as it does not allow for the inference
of all the dependencies
that are valid over the given table. 
Theorem~\ref{theorem:notm} states this,
that an OCD with a common prefix
between \emph{left} and \emph{right}
(repeated attributes)
can hold over a table,
while no OCD without repeated attributes holds.
Our proof of Theorem~\ref{theorem:notm}
is by example,
offering a simple counter-example to the completeness premise
in \cite{CSM19}.



\begin{theorem}\label{theorem:notm}
\emph{
$\R{R}$ $\not \models$ 
$\simular{\lst{Y}}{\lst{Z}}$, 
$\R{R}$ $\not \models$ 
$\simular{\lst{XY}}{\lst{Z}}$ and
$\R{R}$ $\not \models$ 
$\simular{\lst{Y}}{\lst{XZ}}$ 
do not imply $\R{R}$ $\not \models$ $\simular{\lst{XY}}{\lst{XZ}}$
}
\end{theorem}

\begin{proof}
It suffices to construct a table in which the OCD
of the form

\begin{itemize}[nolistsep]
\item \(\simular{\lst{XY}}{\lst{XZ}}\)
\end{itemize}

\noindent
holds, but OCDs

\begin{itemize}[nolistsep]
\item $\simular{\lst{Y}}{\lst{Z}}$,
\item $\simular{\lst{XY}}{\lst{Z}}$, and
\item $\simular{\lst{Y}}{\lst{XZ}}$
\end{itemize}

\noindent
do not.

Consider Table~\ref{table:sample} constructed over attributes $\A{A}$, $\A{B}$ and $\A{C}$. 
In Table~\ref{table:sample},
an OCD $\simular{\Lst{\A{A},\A{B}}}{\Lst{\A{A},\A{C}}}$ holds,
but $\simular{\Lst{\A{B}}}{\Lst{\A{C}}}$,
$\simular{\Lst{\A{AB}}}{\Lst{\A{C}}}$, and 
$\simular{\Lst{\A{B}}}{\Lst{\A{AC}}}$ do not. 
\end{proof}

In \cite{CSM19},
the authors only show---%
as is stated in Theorem~\ref{theorem:ifm} below---%
that OCDs of the form $\simular{\lst{XY}}{\lst{XZ}}$ can be derived from $\simular{\lst{Y}}{\lst{Z}}$ 
(Theorem 3.5 via Theorem 3.10 in \cite{CSM19}).


\begin{theorem}\label{theorem:ifm}
\emph{~\cite{CSM19}}
\emph{%
    If $\R{R}$ $\models$ 
    $\simular{\lst{Y}}{\lst{Z}}$, then $\R{R}$ $\models$ $\simular{\lst{XY}}{\lst{XZ}}$%
}
\end{theorem}

\begin{table}[t]
\begin{center}
\caption{Showing incompleteness.}
\label{table:sample}
    \begin{tabular}{|c||c|c|c|c|c|c|c|c|c|}
        \hline
        \A{\#} & \A{A} & \A{B} & \A{C} \Tstrut\\
        \hline \hline
        $\tup{t1}$ & 0 &  0 &  1  \Tstrut\\
        \hline
        $\tup{t2}$ & 1 &  1 &  0  \Tstrut\\
        \hline
         $\tup{t3}$ & 2 &  3 &  2  \Tstrut\\
        \hline
         $\tup{t4}$ & 3 &  2 &  3  \Tstrut\\
        \hline
    \end{tabular}
\end{center}
\end{table}

Theorem \ref{theorem:ifm} \cite{CSM19} is \emph{true}.
The flaw in their logic is that this theorem
proves only \emph{one} direction
(the ``if'' of an intended ``if and only if'').
The ``only if'' (not proved by the theorem) is implicitly
assumed as \emph{true}, though
(while it assuredly is not).
It follows that
their claim of canonical completeness 
for their definition of minimal OCDs is incorrect
(Section 3.3 in \cite{CSM19}).
OCDs with common prefixes between its \emph{left} and \emph{right}
attribute lists
are \emph{not} redundant,
by Theorem~\ref{theorem:notm}.
This leads to an \emph{incomplete} approach for OD discovery,
as the recovery of the full set of valid dependencies
is not possible. 

Details of the OD discovery algorithm, OCDDISCOVER,
by Consonni et al.~\cite{CSM19} are presented
in their Section 4.
Let $\set{U}$ be a set of attributes
over a relation schema $\R{R}$.
In the first level of the lattice,
they generate candidates
of the form $\simular{\A{A}}{\A{B}}$,
where $\A{A}, \A{B} \in \set{U}$ and $\A{A}$ $\not =$ $\A{B}$.
(An OCD $\simular{\A{B}}{\A{A}}$ is not generated
as it is equivalent to $\simular{\A{A}}{\A{B}}$.)
At each level
of the lattice (Fig.~\ref{figure:lattice2}),
if the candidate $\simular{\lst{X}}{\lst{Y}}$ is
order compatible,
they generate dependencies for the next level of the lattice.
For each attribute not already present in the OCD,
for each attribute
$\A{A} \in \set{U} \setminus \{ \set{X} \cup \set{Y} \}$,
they add it to the right of each attribute list;
i.e.,
$\simular{\lst{X}\A{A}}{\lst{Y}}$ and $\simular{\lst{X}}{\lst{Y}\A{A}}$.
Thus,
important OCDs with repeated attributes
in a common prefix
are never considered
(as is consistent with their
incorrect definition of minimality for OCDs).
For example, an OCD $\simular{ \Lst{\AB{year}, \AB{month}} }{ \Lst{ \AB{year}, \AB{week}} }$
would be missed. As a consequence, the authors do not discover ODs with repeated attributes, such as $\orders{ \Lst{ \AB{year, salary} } }{ \Lst{ \AB{year}, \AB{bin} } }$ (recall Table~\ref{table:taxes}).

\begin{figure}[t]
    \center
    \includegraphics[scale=.99]{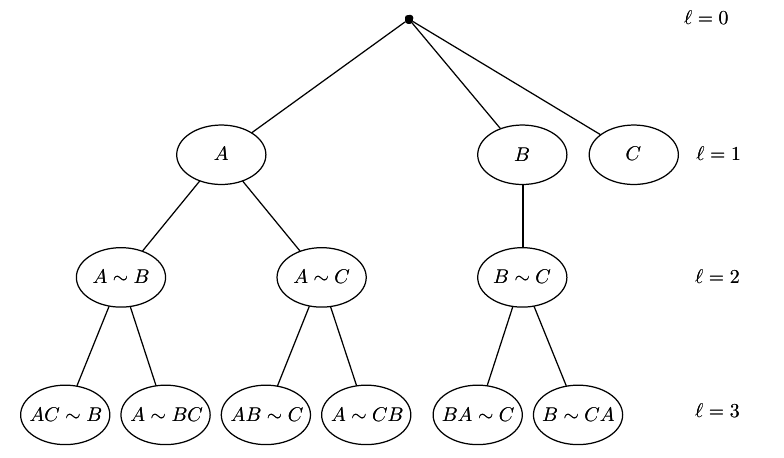}
        \vskip -0.1cm
     \caption{Lattice permutation tree.}
     \label{figure:lattice2}
        \vskip -0.1cm
\end{figure}

In contrast,
our FASTOD algorithm~\cite{SGG+17,SGG+18} is \emph{complete}
for OD discovery.
It does not miss dependencies with common prefixes.
This is because the algorithm considers
as candidates
dependencies of the set-based form: OCD $\{ \set{X} \}$: $\simular{\A{A}}{\A{B}}$.
Thus,
dependencies with common prefixes
are considered.
(This is built into the \emph{context}, set-based notation
used in~\cite{SGG+17,SGG+18},
and cannot be missed when using this representation.)

%% file: exp.tex
\section{Experimental Analysis}\label{section:evaluation}

We demonstrate
that the experimental analysis
in Consonni et al.~\cite{CSM19}
that compares their OD discovery algorithm, OCDDISCOVER,
with ours, FASTOD \cite{SGG+17,SGG+18},
is incorrect.
The authors misinterpret the set-based canonical representation
for ODs as introduced in~\cite{SGG+17,SGG+18}
and as used in FASTOD.
They conflate OCDs and ODs as we report them
when reporting the results.
In \cite{SGG+17,SGG+18},
we report the numbers of found FDs and OCDs.
In \cite{CSM19},
they incorrectly report these as the FDs and ODs, respectively, 
that we found.
This occurs in their Table 6,
where, for instance,
they report 400 ODs and 89,571 FDs found by FASTOD,
whereas this should be 400 OCDs and 89,571 FDs, respectively.

As a consequence of this misunderstanding
of the set-based canonical representation
for ODs~\cite{SGG+17,SGG+18},
the authors in \cite{CSM19} claim
that the implementation of FASTOD finds ODs
that are not present in the data.
As an example of this,
they provide the OD $\orders{\Lst{\A{B}}}{\Lst{\A{A},\A{C}}}$
over Table~\ref{table:bug}~\cite{CSM19}.
However,
the FASTOD algorithm implementation in question finds
the following ODs with respect to Table~\ref{table:bug},
where clearly the OD $\orders{\Lst{\A{B}}}{\Lst{\A{A},\A{C}}}$ is
not present.

\begin{table}[t]
\begin{center}
\caption{Verifying correctness of implementation.}
\label{table:bug}
    \begin{tabular}{|c|c|c|c|c|}
        \hline
        \A{\#} & \A{A} & \A{B} & \A{C} & \A{D} \Tstrut\\
        \hline \hline
        $\tup{t1}$ & 1 &  3 &  1 & 1  \Tstrut\\
        \hline
        $\tup{t2}$ & 2 &  3 &  3 & 2  \Tstrut\\
        \hline
         $\tup{t3}$ & 2 &  3 &  2 & 2  \Tstrut\\
        \hline
        $\tup{t4}$ & 2 &  5 &  2 & 2  \Tstrut\\
        \hline
        $\tup{t5}$ & 3 &  1 &  2 & 3  \Tstrut\\
        \hline
        $\tup{t6}$ & 4 &  4 &  4 & 2  \Tstrut\\
        \hline
         $\tup{t7}$ & 4 &  5 &  3 & 2  \Tstrut\\
        \hline
    \end{tabular}
\end{center}
\vskip -0.4cm
\end{table}

\vskip +0.2cm
\begin{enumerate}[nolistsep]
    \item OCD $\{ \A{D} \}$: $\simular{\A{A}}{\A{C}}$
    \item OCD $\{ \A{C} \}$: $\simular{\A{A}}{\A{D}}$
    \item FD $\brac{ \A{A} }$$:$$\orders{\emptyLst{}}{\A{D}}$
    \item OCD $\{ \A{B} \}$: $\simular{\A{A}}{\A{D}}$
    \item OCD $\{ \A{B} \}$: $\simular{\A{C}}{\A{D}}$
    \item OCD $\{ \A{B} \}$: $\simular{\A{A}}{\A{C}}$
    \item FD $\brac{ \A{B}, \A{C} }$$:$$\orders{\emptyLst{}}{\A{D}}$
    \item FD $\brac{ \A{B}, \A{C} }$$:$$\orders{\emptyLst{}}{\A{A}}$
    \item FD $\brac{ \A{A}, \A{B} }$$:$$\orders{\emptyLst{}}{\A{C}}$
    \item OCD $\{ \A{C}, \A{D} \}$: $\simular{\A{A}}{\A{B}}$
\end{enumerate}
\vskip +0.2cm

The authors confuse the OCD
$\{ \A{B} \}$: $\simular{\A{A}}{\A{C}}$
with the OD $\orders{\Lst{\A{B}}}{\Lst{\A{A},\A{C}}}$.
Consequently,
they falsely assert that the reason the number of ODs found
by OCDDISCOVER and FASTOD differ is
due to an implementation error
in the implementation of FASTOD that we provided them.%
\footnote{%
    While Consonni et al.~\cite{CSM19} state
    that they were not able to isolate and resolve
    the root cause of what they felt was incorrect behavior
    in the implementation of FASTOD
    (which we had provided to them
    at their request for ``ensuring fairness and reproducibility''),
    they never contacted us to help resolve it.
}
The real reason that the number of reported dependencies differ,
however,
is, of course, that OCDDISCOVER \cite{CSM19} is \emph{incomplete}.
The claim that they outperform the state-of-art despite a much worst asymptotic complexity,
when tested in practice on real datasets, is invalid.

The authors in Consonni et al.~\cite{CSM19}\cite{CSM19} also state that FASTOD considers
all columns to be of type string,
while their code also considers real and integer numbers.
While a minor point,
we wish to clarify that the implementation
we sent the authors does discover ODs
over data types including real and integer numbers.
The dependencies 1--10 reported
in Table~\ref{table:bug} remain the same,
regardless of using numerical or string data type,
given that the values are in the range of 1 to 5. 




%% file: concl.tex
\section{Conclusions}\label{section:conc}
In this article, we have conducted a detailed analysis of the correctness of the results in the recent article by Consonni et al.~\cite{CSM19} concerning the order dependency discovery problem. We have shown that, for the main claimed results related to the OD discovery problem, there are fundamental errors and omissions in the proof or experiments. 
